\documentclass[unabridged,copyright,creativecommons]{eptcs_modified}
\providecommand{\event}{QPL 2014} 
\usepackage{breakurl}             

\usepackage{graphicx}
\usepackage{amsmath, amssymb}
\usepackage{bussproofs}
\usepackage{enumerate}
\usepackage{verbatim}
\usepackage{lscape}
\usepackage{multirow}
\usepackage{wrapfig}

\usepackage{amsthm, thmtools}
\usepackage{thm-restate}
\usepackage{nameref}
\usepackage{cleveref}

\usepackage[usenames,dvipsnames]{pstricks}
\usepackage{pst-node, pst-tree}
\usepackage{epsfig}
\usepackage{pst-grad} 
\usepackage{pst-plot} 

\usepackage{multicol}

\theoremstyle{plain} 
\declaretheorem[numberwithin=section]{theorem}
\declaretheorem[numberwithin=section]{corollary}
\declaretheorem[numberwithin=section]{lemma}
\theoremstyle{definition}
\declaretheorem[style=definition,numberwithin=section]{definition}
\declaretheorem[style=definition]{example}
\declaretheorem[style=definition]{remark}

\newcommand{\turnaround}[1]{%
  \rotatebox[origin=c]{180}{\ensuremath#1}}

\newcommand{\logicpar}{\turnaround{\&}}

\title{The dagger lambda calculus}
\author
{
    Philip Atzemoglou
    \institute
    {
        Department of Computer Science,\\
        University of Oxford,\\
        Oxford, UK
    }
    \email
    {
        \includegraphics{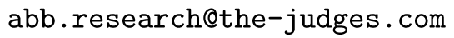}
    }
}
\def\titlerunning{The dagger lambda calculus}
\def\authorrunning{Philip Atzemoglou}

\begin{document}
\maketitle

\begin{abstract}
    We present a novel lambda calculus that casts the categorical approach to the study of quantum protocols \cite{AC04} into the rich and well established tradition of type theory. Our construction extends the linear typed lambda calculus \cite{AT10} with a linear negation \cite{Abr93} of "trivialised" De Morgan duality \cite{AD06}. Reduction is realised through explicit substitution, based on a symmetric notion of binding of global scope, with rules acting on the entire typing judgement instead of on a specific subterm. Proofs of subject reduction, confluence, strong normalisation and consistency are provided, and the language is shown to be an internal language for dagger compact categories.
\end{abstract}

\section{Introduction}
\label{Section:Introduction}
\subsection{Motivation}
\label{Subsection:Motivation}
Since the turn of the century, the study of quantum protocols and quantum computation has gained new momentum through the introduction of a category theoretic approach in the works of \cite{AC04} and \cite{Sel05}. This approach has primarily been using dagger compact categories. In addition to introducing categories to the study of quantum computation, however, the line of work that sprang from this approach has been instrumental in driving a new breed of diagrammatic calculi \cite{CP06, CP07, CPP10, CPP08, CD08, CD11}.

In parallel to this approach, another very prominent line of research was seen in the works of \cite{Sel04a, Sel04b, vTD03, vT04, SV06, SV08, SV10} and was geared towards the development of a quantum programming language. This approach was seminal in establishing a semantic approach to quantum programming language design and focused primarily in designing a higher order lambda calculus for quantum computation. More specifically, in \cite{SV10}, a quantum lambda calculus with a complicated set of rules is presented, whose structural equations nevertheless allow for higher-order structures. The rest of the work towards constructing a concrete model for the language's semantics remains an open problem.

The purpose of this paper is to bridge these two approaches, bringing the programming languages approach closer to the categorical approach, by casting the diagrammatic formalism into the rich and well established tradition of type theory.

\subsection{Summary of results}
\label{Subsection:Summary of results}
Since Symmetric Monoidal Closed categories are the precursor to Compact Closed and Dagger Compact categories, we begin our construction by extending the linear typed lambda calculus of \cite{AT10}. Similarly to the approach used by \cite{Abr93}, we introduce a linear negation operator. Contrary to \cite{Abr93}, however, because \textit{quantum logics} equate $\otimes$ with $\logicpar$ \cite{AD06}, our linear negation operator only allows for a "trivialised" form of De Morgan duality. We also redefine the notion of binding, as a symmetric relation whose scope spans the entire sequent. Reduction works by means of an explicit substitution, in the spirit of the operational semantics of the linear chemical abstract machine \cite{Abr93}. The rules for explicit substitution act globally on the entire typing judgement, instead of limiting their scope to a specific subterm.

By designing our calculus in this way we manage to deconstruct lambda abstraction, one of the traditional primitives of computation, into finer notions of tensor-based binding. This allows us to easily reason with binding operations, such as teleportation, even when they are performed on compound terms. The representation of those operations remains the same, regardless of whether they are teleporting a state or an entire function. A detailed example of this is presented in the end of the Appendix.

An elimination procedure allows us to reconstruct Application using Cut, hence removing it from our primitive rule set. The new rules allow for a fully symmetric language, where inputs and outputs are treated as elements of a symmetric relation, and give rise to a new structural rule called the \textit{dagger-flip}. The resulting set of rules is minimal and simple to use, which allows us to easily prove the properties of subject reduction, confluence, strong normalisation and consistency. Our analysis of the language's semantics is completed by a proof that the dagger lambda calculus is an internal language for dagger compact categories.

\section{The dagger lambda calculus}
\label{Section:The dagger lambda calculus}
Dagger compact categories were first introduced in \cite{ABP99}, albeit under a different name, using some of the terminology of \cite{DR89}. They were later proposed by \cite{AC04} and \cite{Sel05} as an axiomatic framework for the study of quantum protocols. Though a lot of work has been done on categorically driven quantum programming languages \cite{SV06}, \cite{SV08} and \cite{SV10}, these lambda calculi did not provide a way of modelling the dagger functor of dagger compact categories. The work of \cite{BS10} highlighted the importance of dagger compact categories for the semantics of quantum computation; it presented a rough correspondence between quantum computation, logic and the lambda calculus, yet its type theory fell short of providing a correspondence to the entire structure of dagger compact categories. This section fills this gap by presenting the \textit{dagger lambda calculus}: a computational interpretation for dagger compact categories.

\subsection{Language construction}
\label{Subsection:Language construction}
We will now construct a language for \textit{dagger compact categories} by defining well formed formulas for terms, types and sequents. The rules for deriving these formulas will be given in the form of Gentzen-style inference rules. In order to give computational meaning to our language, we will reformalise the typing dynamics of the linear typed lambda calculus \cite{AT10} with the explicit substitution of the linear chemical abstract machine \cite{Abr93}. The linear negation we will use causes a significant collapse between conjunction and disjunction, extends tensor to a (potentially) binding operator, and provides us with a semantics similar to that of the proof nets in \cite{AD06}. The set of rules is kept at a minimum, allowing for clean proofs of the various desired properties. Many familiar computational notions do not appear as primitives, but they do arise as constructed notions in good time.

\begin{definition}[Variables, constants and terms in the dagger lambda calculus]
The fundamental building blocks of our language are \textit{variables}; they are denoted by single letters and are traditionally represented using the later letters of the alphabet (i.e. $x,y,z$). We also allow for the use of \textit{constant terms} (i.e. $c_1,c_2,c_3$); these are terms with an inherent value and cannot serve as placeholders for substitution. These primitives can then be combined with each other to form composite \textit{terms}, denoted by different combinations of the following forms:
\[ \langle term \rangle \;\;\;\; ::= \;\;\;\; variable \;\; | \;\; \langle term \rangle_* \;\; | \;\; \langle term \rangle \otimes \langle term \rangle \;\; | \;\; constant \]
\end{definition}

\begin{definition}[Types in the dagger lambda calculus]
Every term in our language, regardless of whether it is a variable, a constant or composite, has a \textit{type}. We will first start by defining a set of \textit{atomic types}; these are traditionally represented using capital letters (i.e. $A,B,C$). Atomic types can then be combined to give us types of the following forms:
\[ \langle type \rangle \;\;\;\; ::= \;\;\;\; atomic \;\;\; | \;\;\; \langle type \rangle^* \;\;\; | \;\;\; \langle type \rangle \otimes \langle type \rangle \]
\end{definition}

The star operator that we use is not a repetition operator; instead, it corresponds to a particular form of \textit{linear negation}. As one would expect from a negation operation, the star operator is involutive $(a_*)_* \equiv a$ and $(A^*)^* \equiv A$. Abramsky \cite{Abr93} proposed using linear negation as the passageway between Intuitionistic Linear Logic and Classical Linear Logic. The linear negation used in \cite{AD06} "trivialized" the notion of De Morgan duality of \cite{Abr93} by setting $(A \otimes B)^* := A^* \otimes B^*$. The linear negation that we use is similar to the one used in \cite{CPP08}; it distributes differently over tensor by performing a swap of the terms/types at hand and allows for a more "planar" representation. An exchange rule, presented later in this section, will maintain the symmetry of the language's tensors.
\begin{definition}[Linear negation]
The star operator is a form of linear negation whose De Morgan duality is defined by: $(a \otimes b)_* := b_* \otimes a_*$ on terms and $(A \otimes B)^* := B^* \otimes A^*$ on types.
\end{definition}

\begin{definition}[Scalars]
One of the language's atomic types, denoted by $I$, acts as the tensor unit. One of the very important properties of the type $I$ is \textit{negation invariance}, whereby $I \equiv I^*$. We say that a term $i$ is a \textit{scalar} iff it is of type $I$.
\end{definition}

\begin{definition}[Dimensions]
For every type $A$, we will define a scalar constant $D_A : I$, referring to it as the \textit{dimension} of type $A$. The dimension of $I$ is defined to be $D_I = 1 : I$, where $1 = 1_* : I \equiv I^*$.
\end{definition}

\begin{definition}[Soup connection]
A \textit{soup connection} is an ordered pair of equityped terms. A soup connection between two terms of type $A$ is written as $t_1 :_A t_2$ and is an element of the cartesian product of the terms of type $A$ with themselves. To simplify our notation, we write the connection as $t_1 : t_2$, omitting the type, whenever there is no ambiguity about the type of the connected terms. Soup connections do not form a symmetric relation; we use the property $a_1:a_2 \equiv a_{2*}:a_{1*}$ to equate some soup terms by collapsing them into the same congruence class. Moreover, soup connections are not self-dual; we define a \textit{negation} on soup connections as $(t:u)_* := t_*:u_* \equiv u:t$.
\end{definition}

\begin{definition}[Soup]
A \textit{soup} is a set of soup connections, where not all of the connections have to be of the same type. The resulting soup is of the form $S = \{v_1:v_2, \ldots, v_{m-1}:v_m\}$. All of the computation in our language is performed inside the relational soup, by treating its constituent soup connections as a form of explicit substitution. Our negation extends naturally into a \textit{soup negation} whereby $(S \cup S^\prime)_* := S_* \cup S^\prime_*$.
\end{definition}

\begin{definition}[Typing judgements in the dagger lambda calculus]
The \textit{typing judgements}, or \textit{sequents}, of our language are composed of terms, their respective types, and a relational soup. A typing judgement is thus represented by:
\[ t_1:A_1, \; t_2:A_2, \; \ldots, \; t_n:A_n \vdash_S t:B \]
\end{definition}

\begin{example}
In the following typing judgement, the types of $t_1$ and $t_2$ are both known to be $A$. Similarly, we know that both $D_C$ and $1$ are scalars, so their type is $I$. We omit writing the types for soup connections $t_1:t_2$ and $D_C:1$ but, to prevent ambiguity, we have to write it for $x :_B x$, because we have no other way of deducing it from the sequent:
\[ t_1:A \vdash_{\{ t_1 : t_2, x :_B x, D_C : 1 \}} t_2 : A \]
\end{example}

Now that we know which formulas are well formed in our language, we can proceed by defining a notion of binding. Contrary to what we are used to from the lambda calculus, where the notion of binding is restricted in scope to the confines of a single term, the dagger lambda calculus supports a binding that is global and whose scope spans the entire typing judgement. The computational interpretation of classical linear logic, which was provided by \cite{Abr93} in his linear chemical abstract machine, views two occurrences of the same variable as two ends of a communication channel. Adhering to the spirit of that definition, we define binding as follows:
\begin{definition}[Bound variables and terms in the dagger lambda calculus]
For any variable $x$, we say that it is a \textit{bound variable} when it appears twice within a given sequent, regardless of where in the sequent those instances appear. We can also say that an instance of that variable is \textit{captured} by the other instance of the variable in the sequent. As such, variable capture is not limited to the scope of a single term but spans the entire sequent. For any term $t$ that does not contain any occurrences of constants, we say that that term is captured when it consists entirely of variables that are captured within the scope of the current sequent. We use the phrases \textit{bound term} and \textit{bundle of bound variables} interchangeably when referring to captured terms. Trivially, a bound variable is also a bound term.
\end{definition}
\begin{example}
In the following sequent, $x_1$, $x_2$, $y_1$, $y_2$ and $f$ are all bound variables. The individual variables may be free when looking at subterms $x_{1*} \otimes y_1$ and $x_{2*} \otimes y_2$ but, when considering the scope of the entire sequent, they are captured by other occurrences of themselves in the soup. Moreover, the terms $x_{1*} \otimes y_1$ and $x_{2*} \otimes y_2$ are both bound terms because they contain no constants and they consist solely of variables that are captured by variables in the soup:
\[ x_{1*} \otimes y_1 : A^* \otimes B \vdash_{\{x_{1*} \otimes y_1:f, f:x_{2*} \otimes y_2\}} x_{2*} \otimes y_2 : A^* \otimes B \]

In the following sequent, $f$, $y$, $x_1$ and $x_2$ are bound variables; they can also be viewed as bound terms since a single variable is a term and since they do not contain any constants. The term $x_{1}* \otimes x_1$ is a bundle of bound variables because it contains no constants and consists solely of bound variables. The term $c_* \otimes x_2$, however, is not a bundle of bound variables because it contains a constant called $c$:
\[ f : A^* \otimes B \vdash_{\{ x_{1}* \otimes x_1:c_* \otimes x_2, f:x_{2*} \otimes y \}} y : B \]
\end{example}
\begin{remark}
As will become obvious from our language's sequent rules, which will impose linearity constraints on the introduction of variables, the nature of linearity in our language mandates that all of the variables within a given sequent occur exactly twice. This means that all of the free variables in a given term will occur once more in the sequent within which they reside, hence becoming captured in the scope of that sequent. Within that scope, all terms will essentially consist of captured variables and constants.
\end{remark}

\begin{definition}[$\alpha$-renaming on variables in the dagger lambda calculus]
A bound variable $x$ can be $\alpha$-renamed by replacing all of its instances, in a given sequent, with a bundle of bound variables $t$. The term $t$ has to be of the same type as $x$, must not contain any constants (since it will be a bundle of bound variables), and it must consist of variables that do not already appear in the sequent.
\end{definition}

\noindent We can now extend the operation of $\alpha$-renaming to operate on captured terms:
\begin{definition}[$\alpha$-renaming on terms in the dagger lambda calculus]
A bound term $t$ can be $\alpha$-renamed by either $\alpha$-renaming its constituent variables or, in cases where $t$ appears twice in a given sequent, by replacing all of its instances with a variable $x$. The variable $x$ has to be of the same type as $t$ and it must not already appear in the sequent.
\end{definition}

\begin{definition}[$\alpha$-equivalence in the dagger lambda calculus]
We define a notion of \textit{$\alpha$-equivalence} as the reflexive, symmetric and transitive closure of $\alpha$-renaming. In other words, we say that two sequents are \textit{$\alpha$-equivalent}, or \textit{equivalent up to $\alpha$-renaming}, when one can be transformed to the other by $\alpha$-renaming zero or more terms.
\end{definition}

\begin{example}
Going back to the examples we used earlier, the sequent
\[ x_{1*} \otimes y_1 : A^* \otimes B \vdash_{\{x_{1*} \otimes y_1:f, f:x_{2*} \otimes y_2\}} x_{2*} \otimes y_2 : A^* \otimes B \]
\noindent is $\alpha$-equivalent to
\[ g : A^* \otimes B \vdash_{\{g:f, f:x_{2*} \otimes y_2\}} x_{2*} \otimes y_2 : A^* \otimes B \]
\noindent because we can $\alpha$-rename the bound term $x_{1*} \otimes y_1$ into the variable $g$. Similarly, the sequent
\[ f : A^* \otimes B \vdash_{\{ x_{1}* \otimes x_1:c_* \otimes x_2, f:x_{2*} \otimes y \}} y : B \]
\noindent is $\alpha$-equivalent to
\[ x_{3*} \otimes y_2 : A^* \otimes B \vdash_{\{ x_{1}* \otimes x_1:c_* \otimes x_2, x_{3*} \otimes y_2:x_{2*} \otimes y_1 \}} y_1 : B \]
\noindent because we can $\alpha$-rename the bound variable $y$ into $y_1$ and also $\alpha$-rename the bound variable $f$ into the term $x_{3*} \otimes y_2$.
\end{example}

\begin{definition}[Typing contexts in the dagger lambda calculus]
The left-hand-side of a typing judgement is actually a list of typed terms. We use the letters $\Gamma$ and $\Delta$ as shorthand for arbitrary (possibly empty) lists of such terms. Let $\Delta$ be the list $t_1:T_1, t_2:T_2, \ldots, t_n:T_n$. We define $\bigotimes\Delta$ to be the term $\left(\left(\left(t_1 \otimes t_2\right) \otimes \ldots\right) \otimes t_n\right) : \left(\left(\left(T_1 \otimes T_2\right) \otimes \ldots\right) \otimes T_n\right)$, referring to it as \textit{$\Delta$ in tensor form}.
\end{definition}

Our language exposition features a Gentzen-style Sequent Calculus, which provides us with the inference rules used to produce judgements. Rules with a double line are bidirectional; sequents matching the top of the rule can be used to derive sequents matching the bottom and vice versa. The rules are formed in a way that allows composite terms to appear to the left of the turnstile. The sequent rules are:
\vskip 0.3cm
\begin{minipage}{.4\linewidth}
\begin{flushright}
    \AxiomC{}
    \RightLabel{\small Id,}
    \UnaryInfC{$x:A \vdash x:A$}
    \DisplayProof \vskip 0.4cm

    \AxiomC{$\Gamma \vdash_{S_1} a:A$}
    \AxiomC{$a^\prime : A, \Delta \vdash_{S_2} b:B$}
    \RightLabel{\small Cut$^*$,}
    \BinaryInfC{$\Gamma, \Delta \vdash_{S_1 \cup S_2 \cup \{a:a^\prime\}} b:B$}
    \DisplayProof \vskip 0.4cm

    \AxiomC{$\Gamma \vdash_{S_1} a:A$}
    \AxiomC{$\Delta \vdash_{S_2} b:B$}
    \RightLabel{\small $\otimes R^*$,}
    \BinaryInfC{$\Gamma, \bigotimes\Delta \vdash_{S_1 \cup S_2} a \otimes b : A \otimes B$}
    \DisplayProof
\end{flushright}
\end{minipage}
\begin{minipage}{.3\linewidth}
\begin{flushright}
    \AxiomC{$a:A \vdash_S b:B$}
    \RightLabel{\small Negation,}
    \UnaryInfC{$a_*:A^* \vdash_{S_*} b_*:B^*$}
    \DisplayProof \vskip 0.4cm

    \AxiomC{$a:A, \Gamma \vdash_S b:B$}
    \doubleLine\RightLabel{\small Curry,}
    \UnaryInfC{$\Gamma \vdash_S a_* \otimes b : A^* \otimes B$}
    \DisplayProof \vskip 0.4cm

    \AxiomC{$\Gamma, a:A, b:B \vdash_S c:C$}
    \doubleLine\RightLabel{\small $\otimes L$.}
    \UnaryInfC{$\Gamma, a \otimes b : A \otimes B \vdash_S c:C$}
    \DisplayProof
\end{flushright}
\end{minipage} \hskip 0.1cm
\begin{minipage}{.01\linewidth}
\rule{0.4pt}{4.6cm}
\end{minipage} \hskip 0.1cm
\begin{minipage}{.2\linewidth}
    \begin{footnotesize}
        $^*$: The sequents merged by the Cut and $\otimes R$ rules must not share any common variables. Whenever we use these two rules on sequents whose variables overlap, we have to $\alpha$-rename them first to prevent capturing the variables.
    \end{footnotesize}
\end{minipage}
\begin{remark}
The identity axiom (Id) is the only inference rule we have for introducing variables into our expressions. Consequently, variables are always introduced as bound pairs. The Cut rule establishes a connection between the output of one sequent and the input of another. The $\otimes R$ rule tensors two sequents together, preserving tensor associativity by turning $\Delta$ into $\bigotimes\Delta$. Given the capturing restriction for Cut and $\otimes R$, no other bindings can be introduced in our expressions. As such, variables will appear exactly twice in a sequent. We call this property \textit{linearity}, the sequents \textit{linear}, and the restrictions on Cut and $\otimes R$ \textit{linearity constraints}.
\end{remark}

We sometimes use sequents with an empty right-hand-side, for instance $a:A, \Gamma \vdash$ as shorthand for $a:A, \Gamma \vdash 1:I$. Such sequents are easy to produce by using \textit{Uncurrying}, the inverse of the \textit{Curry} rule, together with the constant $1:I$:
\begin{prooftree}
\AxiomC{$\Gamma \vdash a_*:A^*$}
\AxiomC{$\vdash 1:I$}
\RightLabel{\small $\otimes R$}
\BinaryInfC{$\Gamma \vdash a_* \otimes 1 : A^* \otimes I$}
\RightLabel{\small Uncurry}
\UnaryInfC{$a:A, \Gamma \vdash 1:I$}
\end{prooftree}

The language has a structural exchange rule that can be used to swap terms on the left hand side of a sequent. When navigating through a proof tree, instances of the exchange rule can be used to keep track of which terms were swapped and at which points during a derivation:
\begin{prooftree}
\AxiomC{$\Gamma, a:A, b:B, \Delta \vdash c:C$}
\doubleLine\RightLabel{\small Exchange.}
\UnaryInfC{$\Gamma, b:B, a:A, \Delta \vdash c:C$}
\end{prooftree}

Our language also has two unit rules, $\lambda_\Gamma$ and $\rho_\Gamma$, that are used to more accurately represent scalars:
\begin{center}
\AxiomC{$\Gamma \vdash_{S \cup \{i_*:1\}} b:B$}
\doubleLine\RightLabel{\small $\lambda_\Gamma$}
\UnaryInfC{$i:I, \Gamma \vdash_S b:B$}
\DisplayProof, \hskip 0.5cm
\AxiomC{$\Gamma \vdash_{S \cup \{i_*:1\}} b:B$}
\doubleLine\RightLabel{\small $\rho_\Gamma$}
\UnaryInfC{$\Gamma, i:I \vdash_S b:B$}
\DisplayProof.
\end{center}

Our language dynamics are defined through soup rules. These rules explain how the relational connections propagate within the soup, giving rise to an operational semantics for a form of "global substitution" that resembles pattern matching on terms. The soup propagation rules, called \textit{bifunctoriality}, \textit{trace} and \textit{cancellation} respectively, are:
\begin{align*}
    S \cup \{a \otimes b : c \otimes d\}    \longrightarrow & \; S \cup \{a:c, b:d\} \\
    S \cup \{x :_A x\}                      \longrightarrow & \; S \cup \{ D_A : 1 \} \\
    S \cup \{1 : 1\}                        \longrightarrow & \; S
\end{align*}

\noindent where $\psi$ is a constant and $x$ is a variable. Our soup rules also contain a \textit{consumption rule}. This rule uses up a relational connection between $\{t : u\}$ to perform a substitution in the typing judgement. Note, however, that the term we are substituting for has to be one that was captured in the scope of the sequent:
\[
\Gamma \vdash_{S \cup \{t : u\}} b:B  \longrightarrow \bigg(\Gamma \vdash_S b:B\bigg)
\begin{cases}
    [t / u],    & \text{if $u$ does not contain constants}\\
    [u / t],    & \text{if $t$ does not contain constants}
\end{cases}
\]

If $t$ and $u$ are both without constants, linearity implies that their constituent variables were all captured in the scope of the original sequent. In such a case, we can choose the way in which we want to substitute. This gives us a symmetric notion of substitution, where our choice of substitution does not affect the typing judgement, as the sequents will be equivalent up to alpha renaming.

\begin{example}
\noindent Consider the following sequent:
\[ f : A^* \otimes B \vdash_{\{ f:c_* \otimes y \}} y : B \]
\noindent The variable $f$ is captured within the scope of the sequent. As such, we can use the \textit{consumption rule} to consume the connection in the soup and substitute $c_* \otimes y$ for $f$ in the rest of the sequent. This will change the sequent to:
\[ c_* \otimes y : A^* \otimes B \vdash y : B \]

\noindent Alternatively, if we had $\alpha$-renamed the original sequent to:
\[ x_{1*} \otimes y_1 : A^* \otimes B \vdash_{\{ x_{1*} \otimes y_1:c_* \otimes y_2 \}} y_2 : B \]
\noindent we could have then used the \textit{bifunctoriality} rule to split the soup connection:
\[ x_{1*} \otimes y_1 : A^* \otimes B \vdash_{\{ x_{1*}:c_*, y_1:y_2 \}} y_2 : B \]
\noindent The first connection of the resulting soup is only consumable in one way, since $c$ is a constant, by substituting $c_*$ for $x_{1*}$. The second soup connection, however, presents us with a choice, since both $y_1$ and $y_2$ are captured in the sequent. One choice will give us
\[ c_* \otimes y_2 : A^* \otimes B \vdash y_2 : B \]
\noindent while the other choice will give us
\[ c_* \otimes y_1 : A^* \otimes B \vdash y_1 : B \]

\noindent Upon closer inspection, one will notice that all three of the resulting sequents are $\alpha$-equivalent.
\end{example}

\begin{definition}[Soup reduction]
We use the term \textit{soup reduction} to refer to the binary relation that extends $\alpha$-equivalence with the sequent transformations that are caused by applying one of the soup rules. Thus, for two sequents $\Gamma \vdash_{S_1} t : T$ and $\Gamma \vdash_{S_2} t : T$, if the soup $S_1$ is transformed into $S_2$ through the application of one of the soup propagation rules, $S_1 \rightarrow S_2$, then we say that one sequent reduces to the other via \textit{soup reduction}. Similarly, if a sequent $J_1$ is transformed into $J_2$ by using the consumption rule to perform a substitution, we say that $J_1$ reduces to $J_2$ via \textit{soup reduction}.
\end{definition}

\begin{definition}[Soup equivalence]
We define a notion of \textit{soup equivalence} as the reflexive, symmetric, and transitive closure of soup reduction. In other words, we say that two sequents $J_1$ and $J_2$ are \textit{soup-equivalent}, or \textit{equivalent up to soup-reduction}, when we can convert one to the other by using zero or more instances of $\alpha$-renaming and soup reduction.
\end{definition}

\noindent We can now use the rules that we have defined so far in order to express the computational notion of application:
\begin{definition}[Application in the dagger lambda calculus]
Let $t$ and $f$ be terms such that $t:A$ and $f:A^* \otimes B$ for some types $A$ and $B$. We define the \textit{application} $ft$ as a notational shorthand for representing a variable $x:B$, along with a connection in our soup. The origins of the application affect the structure of its corresponding soup connection:
\begin{center}
$ft:B, \Gamma \vdash c:C \;\;   := \;\; x:B, \Gamma \vdash_{\{f:t_* \otimes x\}_*} c:C$ \hskip 0.5cm and \hskip 0.5cm
$\Gamma \vdash ft:B \;\;        := \;\; \Gamma \vdash_{\{f:t_* \otimes x\}} x:B$
\end{center}
For an application originating inside our soup, we have:
\begin{center}
$\{ ft:c \} := \{ x:c \} \cup \{ f:t_* \otimes x \}$ \hskip 0.5cm and \hskip 0.5cm
$\{ c:ft \} := \{ c:x \} \cup \{ f:t_* \otimes x \}_*$
\end{center}
\end{definition}

\begin{corollary}[Beta reduction]
This immediately allows us to represent a form of \textit{beta reduction}. Instead of relying on an implicit meta-concept of substitution, our beta reduction is going to express the binding and reduction of terms by connecting them in the soup by setting $(a_* \otimes b)t \stackrel{\beta}{\longrightarrow} b$, while causing $\{t:a\}$ or $\{t:a\}_*$ to be added to the relational soup.
\end{corollary}
\begin{proof}
This is derived from our definition of application because $(a_* \otimes b)t$ represents a variable $x$ along with one of two possible connections in our soup. The soup connection can be manipulated into:
\[ \{a_* \otimes b : t_* \otimes x\}   \rightarrow \{a_*:t_*, b:x\}    \rightarrow \{t:a\} \cup \{b:x\} \]
\[ \{a_* \otimes b : t_* \otimes x\}_* \rightarrow \{a_*:t_*, b:x\}_*  \rightarrow \{t:a\}_* \cup \{x:b\} \]
The connection between $b$ and $x$ can then be consumed to change the variable $x$ into a $b$. All that remains is $\{t:a\}$ or $\{t:a\}_*$.
\end{proof}

Now that all of the language's rules are in place, we can demonstrate how the familiar notion of lambda abstraction can be reconstructed from the  finer notions of linear negation and tensor, by defining it to be a notational shorthand:

\begin{definition}[Lambda abstraction in the dagger lambda calculus]
Let $\lambda a.b := a_* \otimes b$ and $A \multimap B := A^* \otimes B$
\end{definition}

\noindent The following combinators are used in the rest of this paper:
\begin{align*}
    id_A    & := \lambda a.a \textrm{ (where $a:A$)} & \bar{b} & := \lambda g.\lambda f.\lambda a.g(fa)\\
    \bar{s} & := \lambda (a \otimes b).(b \otimes a) & \bar{t} & := \lambda f.\lambda g.\lambda(x_1 \otimes x_2).(fx_1 \otimes gx_2)
\end{align*}

\begin{theorem}[Admissibility of $\multimap E$]
We can also use the definition of application to demonstrate that an implication elimination rule ($\multimap E$) is admissible within our set of rules.
\end{theorem}
\begin{proof}
    \begin{prooftree}
        \AxiomC{$\Gamma \vdash_{S_1} t:A$}
        \AxiomC{$\Delta \vdash_{S_2} f:A^* \otimes B$}
        \AxiomC{}
        \UnaryInfC{$a:A \vdash a : A$}
        \UnaryInfC{$a_*:A^* \vdash a_*:A^*$}
        \AxiomC{}
        \UnaryInfC{$b:B \vdash b : B$}
        \BinaryInfC{$a_*:A^*, b:B \vdash a_* \otimes b:A^* \otimes B$}
        \UnaryInfC{$a_* \otimes b:A^* \otimes B \vdash a_* \otimes b:A^* \otimes B$}
        \RightLabel{\footnotesize Cut}
        \BinaryInfC{$\Delta \vdash_{S_2 \cup \{f:a_* \otimes b\}} a_* \otimes b:A^* \otimes B$}
        \RightLabel{\footnotesize Uncurry}
        \UnaryInfC{$a:A, \Delta \vdash_{S_2 \cup \{f:a_* \otimes b\}} b:B$}
        \RightLabel{\footnotesize Cut}
        \BinaryInfC{$\Gamma, \Delta \vdash_{S_1 \cup S_2 \cup \{t:a, f:a_* \otimes b\}} b:B$}
        \UnaryInfC{$\Gamma, \Delta \vdash_{S_1 \cup S_2 \cup \{f:t_* \otimes b\}} b:B$}
        \UnaryInfC{$\Gamma, \Delta \vdash_{S_1 \cup S_2} ft:B$}
    \end{prooftree}
\end{proof}

We define some additional notational conventions, so that we can more easily describe the reversal in the causal order of computation:
\begin{definition}[Complex conjugation]
Let $f: A^* \otimes B$ be an arbitrary function. As a notational convention, we set $f^* := \bar{s} f : B \otimes A^*$.
\end{definition}

\begin{theorem}[Admissibility of $\dag$-flip]
We can use the language's rules and definitions in order to admit a new structural rule called the $\dag$-flip. This rule contains all the computational symmetry that we will later need in order to model the \textit{dagger functor}.
\end{theorem}
\begin{proof}
\begin{prooftree}
    \AxiomC{$a:A \vdash_{S} b:B$}
    \RightLabel{\footnotesize Negation}
    \UnaryInfC{$a_*:A^* \vdash_{S_*} b_*:B^*$}
    \RightLabel{\footnotesize Uncurry}
    \UnaryInfC{$b:B, a_*:A^* \vdash_{S_*}$}
    \RightLabel{\footnotesize Exchange}
    \UnaryInfC{$a_*:A^*, b:B \vdash_{S_*}$}
    \RightLabel{\footnotesize Curry}
    \UnaryInfC{$b:B \vdash_{S_*} a:A$}
\end{prooftree}
\end{proof}

\begin{restatable}[Interchangeability of $\dag$-flip and Negation]{lemma}{DaggerFlipNegation}
Alternatively, we could have defined the language by including $\dag$-flip in our initial set of sequent rules. That would have allowed us to admit the Negation rule as a derived rule.
\end{restatable}
\begin{proof}
\begin{prooftree}
    \AxiomC{$a:A \vdash_{S} b:B$}
    \RightLabel{\footnotesize $\dag$-flip}
    \UnaryInfC{$b:B \vdash_{S_*} a:A$}
    \RightLabel{\footnotesize Uncurry}
    \UnaryInfC{$a_*:A^*, b:B \vdash_{S_*}$}
    \RightLabel{\footnotesize Exchange}
    \UnaryInfC{$b:B, a_*:A^* \vdash_{S_*}$}
    \RightLabel{\footnotesize Curry}
    \UnaryInfC{$a_*:A^* \vdash_{S_*} b_*:B^*$}
\end{prooftree}
\end{proof}

\subsection{Scalars}
\label{Subsection:Scalars}
Similarly to the attachable monoid that is described in \cite{Abr05} for multiplying scalars, we can optionally define a multiplication operation for the scalars in the dagger lambda calculus. This is not part of the structure that is necessary to model dagger compact categories computationally, hence the designation \textit{optional}, but it does provide a good example of how connections propagate in the soup:
\begin{definition}[Scalar multiplication]
For any two scalars $m:I$ and $n:I$, we define a multiplication operation $m \cdot n : I$ such that:
\[ m \cdot 1 = 1 \cdot m = m \]
\noindent and
\[ \{ m \cdot p : n \cdot q \} := \{ m:n, p:q \} \]
\end{definition}

\noindent The operation features a number of properties. To help the reader get more accustomed to the way things propagate in the soup, we will demonstrate some of them as an example. First of all, scalar multiplication is \textit{associative}:
\begin{lemma}[Associativity of multiplication]
$ (a \cdot b) \cdot c = a \cdot (b \cdot c) $
\end{lemma}
\begin{proof}
\begin{align*}
    \{ (a \cdot b) \cdot c : 1 \}   & = \{ (a \cdot b) \cdot c : (1 \cdot 1) \cdot 1 \} \\
                                    & = \{ a:1, b:1, c:1 \} \\
                                    & = \{ a \cdot (b \cdot c) : 1 \cdot (1 \cdot 1) \} \\
                                    & = \{ a \cdot (b \cdot c) : 1 \}
\end{align*}
\end{proof}

\noindent The multiplication operation is also \textit{commutative}:
\begin{lemma}[Commutativity of multiplication]
$m \cdot n = n \cdot m$
\end{lemma}
\begin{proof}
\begin{align*}
    \{ m \cdot n : 1 \} & = \{ m \cdot n : 1 \cdot 1 \} \\
                        & = \{ m:1, n:1 \} \\
                        & = \{ n:1, m:1 \} \\
                        & = \{ n \cdot m : 1 \cdot 1 \} \\
                        & = \{ n \cdot m : 1 \}
\end{align*}
\end{proof}

\noindent It is \textit{sesquilinear}:
\begin{lemma}[Sesquilinearity of scalar connections]
$ \{ m : n \} = \{ m \cdot n_* : 1 \} $
\end{lemma}
\begin{proof}
\begin{align*}
    \{ m : n \} & = \{ m \cdot 1 : 1 \cdot n \} \\
                & = \{ m:1, 1:n \} \\
                & = \{ m:1, n_*:1 \} \\
                & = \{ m \cdot n_* : 1 \cdot 1 \} \\
                & = \{ m \cdot n_* : 1 \} \\
\end{align*}
\end{proof}

\noindent Finally, it is easy to deduce that the dimension of a tensor of types distributes into a product of dimensions:
\begin{corollary}[Dimension multiplication]
$ \{ D_A \cdot D_B : 1 \} = \{ D_{A \otimes B} : 1 \} $
\end{corollary}
\begin{proof}
\begin{align*}
    \{ D_A \cdot D_B : 1 \} & = \{ D_A:1, D_B:1 \}                              & = \{ a :_A a, b :_B b \} \\
                            & = \{ a \otimes b :_{A \otimes B} a \otimes b \}   & = \{ D_{A \otimes B} : 1 \}
\end{align*}
\end{proof}

\subsection{Language properties}
\label{Subsection:Language properties}
Our lambda calculus was designed with a minimal set of rules. This has led to a tractable language, where most of the properties are easy to prove by structural induction. Throughout the rest of this section, we establish that our lambda calculus satisfies the following important properties of a calculus: subject reduction, confluence, strong normalisation, and consistency. Sketches of the proofs are provided and more detailed versions can be found in \cite{Atz13}.

\subsubsection{Subject reduction}
\label{Subsubsection:Subject reduction}
The first thing we have to prove, in order to demonstrate that our typing system is well defined, is the consistency of our typing dynamics. In other words, we have to verify that the way in which relational connections propagate through our soup preserves type assignments. This is easy to observe since our soup only connects \textit{equityped} terms. Pair consumption substitutes a term for another of the same type, thus preserving types.

\begin{theorem}[Subject reduction]
Let $J_1$ and $J_2$ be two typing judgements such that $J_1 = \Gamma \vdash_S t_1:A_1$ and $J_2 = \Delta \vdash_{S^\prime} t_2:A_2$. Suppose that these two judgements are such that we can use a soup reduction rule $S \longrightarrow S^\prime$ to reduce one to the other: $J_1 \longrightarrow J_2$. Then, the reduction will not alter type assignments in any way: $types(\Gamma) = types(\Delta)$ and $A_1 \equiv A_2$.
\end{theorem}
\begin{proof}
A longer version of this proof can be found in \cite{Atz13}. The only soup rule that could affect the premises and conclusion of a typing judgement is the consumption rule. The resulting substitution may be global in scope, but it does not affect the sequent's typing, since it is substituting one term for another one of the same type.
\end{proof}

\subsubsection{Normalisation}
\label{Subsubsection:Normalisation}
Strong normalisation is a highly sought after property for lambda calculi, primarily because of the implications it has on the practical implementation of the language. A reduction that is strongly normalising implies that every sequent has a normal form. Furthermore, it requires that the normal form is attained after a finite number of steps, without any chance of running into an infinite reduction loop.

\begin{theorem}[Strong normalisation]
Every sequence of soup reduction steps is finite and ends with a typing judgement that is in normal form.
\end{theorem}
\begin{proof}
A longer version of this proof can be found in \cite{Atz13}, using an induction on the size and structure of the soup reduction. A sequent not in normal form will have a soup with at least one usable connection, for which there are four possible reduction steps. A step using the \textit{trace}, \textit{cancellation} or \textit{consumption} rule will use up that soup connection, the soup being a finite set, leaving us with a smaller usable soup. A step using the \textit{bifunctoriality} rule, bounded in its application by the number of atomic types, will split the soup connection into simpler subtypes.
\end{proof}

\subsubsection{Confluence}
\label{Subsubsection:Confluence}
Another very important property for our language is the Church-Rosser property. It ensures that we can end up with the same sequent regardless of the reduction path we choose to follow. A careful observation of our rewrite rules will reveal that the rules are all left-linear.
\begin{theorem}[Left-linearity]
All of our soup rewrite rules are left-linear.
\end{theorem}
\begin{proof}
In accordance with the linearity constraints of our language, no variable appears more than twice on the left hand side of any of our soup reduction rules.
\end{proof}

One should note, at this point, that our soup rules do exhibit a form of "harmless" overlap. More specifically, the consumption rule ($S \cup \{t:u\} \longrightarrow S$) forms a critical pair with itself in cases where $t$ and $u$ are both bound. Fortunately, as we will see in the next lemma, these pairs prove to be \textit{trivial} as they correspond to sequents that are equivalent up to $\alpha$-renaming.

\begin{theorem}[Symmetry of substitution]
Let $J$ be a typing judgement of the form $J := \Gamma \vdash_{S \cup \{t:u\}} a:A$, where $t$ and $u$ are both bound. The connection $\{t:u\}$ can be consumed in either of two ways; one substitutes $t$ for $u$ and the other substitutes $u$ for $t$ in the typing judgement. Let's call these $J_1$ and $J_2$ respectively. $J_1$ will then be $\alpha$-equivalent to $J_2$.
\end{theorem}
\begin{proof}
Since $t$ and $u$ are both bound, by linearity, we know that they appear exactly once in $\Gamma \vdash_S a:A$. After substitution is performed, $J_1$ will have two occurrences of $t$ where $t$ and $u$ used to be, so $t$ will be a bound term in that judgement. Similarly, $J_2$ will have two occurrences of $u$ where $t$ and $u$ used to be, so $u$ will be a bound term in that judgement. These bound terms occur in the exact same spots, so we can $alpha$-rename $J_1$ to $J_2$ and vice versa.
\end{proof}

\begin{corollary}[No overlap]
The rewrite rules have no overlap up to $\alpha$-equivalence of typing judgements.
\end{corollary}

\begin{theorem}[Confluence]
Our reduction rules have the Church-Rosser property.
\end{theorem}
\begin{proof}
Our set of rewrite rules is \textit{left-linear} and has no significant overlap, since it only gives rise to critical pairs that are \textit{trivial} up to $\alpha$-equivalence. Therefore, our rewrite rules constitute a \textit{weakly orthogonal} rewrite system, which is \textit{weakly confluent} according to \cite{Klo92} (Consider the variation of Theorem 2.1.5 for \textit{weakly orthogonal} TRS's on page 72). Since the rewrite system is both strongly normalising and weakly confluent, we can use Newman's lemma to conclude that it also possesses the Church-Rosser property. See \cite{Klo92} for a more detailed explanation of the properties of orthogonal rewriting systems.
\end{proof}

\subsubsection{Consistency}
\label{Subsubsection:Consistency}
In order to show that our type theory is consistent, we have to show that our soup dynamics do not collapse all equityped terms to the same element.

\begin{theorem}[Consistency]
There exist two terms of the same type, henceforth referred to as $t_1$ and $t_2$, such that $\Gamma \vdash_{S_1} t_1:A$ and $\Gamma \vdash_{S_2} t_2:A$ could never reduce to the same typing judgement.
\end{theorem}
\begin{proof}
Consider two combinators of the same type, $t_1 = id_{A \otimes A}$ and $t_2 = \bar{s}_{A \otimes A}$. Both terms are closed, containing no free variables or constants. The sequents $\vdash id_{A \otimes A} : (A \otimes A) \multimap (A \otimes A)$ and $\vdash \bar{s}_{A \otimes A} : (A \otimes A) \multimap (A \otimes A)$ are distinct normal forms: They are clearly distinct from one another and cannot be further reduced using any of our rules, thereby proving that they could never reduce to the same typing judgement.
\end{proof}

\subsection{Correspondence to dagger compact categories}
\label{Subsection:Correspondence to dagger compact categories}

The purpose of this section is to provide a full Curry-Howard-Lambek correspondence between the dagger lambda calculus and dagger compact categories. We start by defining a directed graph $\mathcal{G}$, representing a signature for dagger compact categories. We then show how that graph can be interpreted to define the free dagger compact category $\mathcal{C}_{Free}$ and the dagger lambda calculus $\dag\lambda$. An appropriate Cut-elimination procedure is defined to partition the sequents of the dagger lambda calculus into equivalence classes up to soup equivalence. The resulting equivalence classes are modular proof invariants represented by denotations. We show that the types and denotations can be used to form a syntactic category, $\mathcal{C}_{Synt}$, and prove that the category is dagger compact. The diagram below, fashioned to resemble the diagram at the bottom of page 49 in \cite{Mac98}, is provided to help visualise the Curry-Howard-Lambek correspondence. In this diagram, $U\mathcal{C}_{Free}$ and $U\mathcal{C}_{Synt}$ are the underlying graphs of their respective categories, where identities, composition, natural isomorphisms and other structural elements of the parent categories have been "forgotten" by applying the forgetful functor $U$. $F$ is the unique functor between the free and the syntactic category, that satisfies the rest of the conditions in the diagram.
\[
\scalebox{1} 
{
\begin{pspicture}(0,-1.7217188)(12.894688,1.7217188)
\usefont{T1}{ptm}{m}{n}
\rput(4.5242186,1.5007813){\large $\mathcal{C}_{Free}$}
\usefont{T1}{ptm}{m}{n}
\rput(4.604219,-1.3192188){\large $\mathcal{C}_{Synt}$}
\usefont{T1}{ptm}{m}{n}
\rput(1.2442187,-1.3192188){\large $\dag\lambda$}
\psline[linewidth=0.04cm,arrowsize=0.05291667cm 2.0,arrowlength=1.4,arrowinset=0.4]{->}(1.576875,-1.2542187)(3.976875,-1.2542187)
\usefont{T1}{ptm}{m}{n}
\rput(2.6482813,-1.5442188){$\ell$}
\psline[linewidth=0.04cm,arrowsize=0.05291667cm 2.0,arrowlength=1.4,arrowinset=0.4]{->}(4.476875,1.2457813)(4.476875,-1.0542188)
\usefont{T1}{ptm}{m}{n}
\rput(4.798281,0.05578125){$!F$}
\usefont{T1}{ptm}{m}{n}
\rput(8.174219,1.5007813){\large $U\mathcal{C}_{Free}$}
\usefont{T1}{ptm}{m}{n}
\rput(8.254219,-1.3192188){\large $U\mathcal{C}_{Synt}$}
\psline[linewidth=0.04cm,arrowsize=0.05291667cm 2.0,arrowlength=1.4,arrowinset=0.4]{->}(7.976875,1.2457813)(7.976875,-1.0542188)
\usefont{T1}{ptm}{m}{n}
\rput(8.368281,0.05578125){$UF$}
\usefont{T1}{ptm}{m}{n}
\rput(11.514218,1.5007813){\large $\mathcal{G}$}
\psline[linewidth=0.04cm,arrowsize=0.05291667cm 2.0,arrowlength=1.4,arrowinset=0.4]{->}(11.276875,1.5457813)(8.776875,1.5457813)
\psline[linewidth=0.04cm,arrowsize=0.05291667cm 2.0,arrowlength=1.4,arrowinset=0.4]{->}(11.276875,1.2457813)(8.276875,-1.0542188)
\end{pspicture}
}
\]

We will prove an equivalence between the free category and the syntactic category. We should note at this point that our typing conventions of an involutive negation ($A \equiv (A^*)^*$) and negation invariance of the tensor unit ($I \equiv I^*$) implicitly introduce equivalence classes on types. Our proof of equivalence will be achieved by fully exhibiting the correspondence in objects and arrows between the two categories, showing that their notions of equality overlap, up to the equivalence classes that are induced by our typing conventions.

\subsubsection{A signature for dagger compact categories}
\label{Subsubsection:A signature for dagger compact categories}
The notion of signature that we will use combines the algebraic signature of \cite{Sel10} with the directed graph used by \cite{Mac98}. Consider a set of object variables $\Sigma_0$. Using the tensor operation, an associated tensor identity, and the duality operator star, we can construct the free $(\otimes, I, \Box^*)$-algebra over $\Sigma_0$. This corresponds to the set of all object terms or vertices in a compact closed category and will be denoted by $Dagger(\Sigma_0)$. Now consider a set $\Sigma_1$ of morphism variables or edges between those vertices. Let $dom, cod$ be a pair of functions such that $dom,cod: \Sigma_1 \longrightarrow Dagger(\Sigma_0)$. Throughout the rest of this section, we will be referring to the graph $\mathcal{G}$ as the directed graph whose vertices and edges are defined by $Dagger(\Sigma_0)$ and $\Sigma_1$. This graph forms the signature upon which we will base both the dagger lambda calculus and our description of the free dagger compact category; it includes all of the symbols but none of the logic of the languages that we want to describe.

\subsubsection{The free dagger compact category}
\label{Subsubsection:The free dagger compact category}
We will now show how to define the free dagger compact category $\mathcal{C}_{Free}$ as an interpretation of the graph $\mathcal{G}$. A highly intuitive introduction to free categories and how they can be generated from directed graphs can be found in \cite{Mac98}. Furthermore, a more extensive presentation of the process of constructing various kinds of free categories can be found in \cite{Sel10}. A more detailed presentation of the incremental buildup to the construction of free dagger compact categories can also be found in \cite{Abr05}.

The set of objects for the free category in this section will be the same as the set of vertices $Dagger(\Sigma_0)$ in the graph $\mathcal{G}$. The set of edges $\Sigma_1$ in the graph is used to generate morphisms for the free category. Thus, an edge of the form $f : A \rightarrow B$ generates an arrow in $\mathcal{C}_{Free}$ which we will denote as $\langle A,f,B \rangle$. The identities are represented by: $\langle A \rangle, \langle B \rangle, \langle C \rangle, \ldots$

The free category over a directed graph, also referred to as a path category, includes morphisms that correspond to the paths generated by combining adjoining edges in $\mathcal{G}$. These morphisms are formed using the free category's composition operation. Given two morphisms $\langle A,f,B \rangle$ and $\langle B,g,C \rangle$, we write their composition in $\mathcal{C}_{Free}$ as $\langle A,f,B,g,C \rangle$.

Since the free category is a monoidal category, it allows us to consider two of the graph's edges concurrently by bringing together their corresponding categorical morphisms using a monoidal tensor product. Given two morphisms $\langle A,f,B \rangle$ and $\langle C,h,D \rangle$, we write their tensor product as $\langle A \otimes C, f \otimes h, B \otimes D \rangle$.

The free category generated by the graph $\mathcal{G}$ also includes a number of morphisms that are part of the dagger compact logical structure. The monoidal natural isomorphisms are written as:
\begin{center}
    $\langle A \otimes (B \otimes C), \alpha_{A,B,C}, (A \otimes B) \otimes C \rangle$ \hskip 2cm
    $\langle I \otimes A, \lambda_A, A \rangle$ \hskip 2cm
    $\langle A \otimes I, \rho_A, A \rangle$
\end{center}

The symmetry isomorphism, and the units and counits are written as:
\begin{center}
    $\langle A \otimes B, \sigma_{A,B}, B \otimes A \rangle$ \hskip 2cm
    $\langle I, \eta_A, A^* \otimes A \rangle$ \hskip 2cm
    $\langle A \otimes A^*, \varepsilon_A, I \rangle$
\end{center}

For every map $\langle A,f,B \rangle$ in the free category, the dagger compact logical structure contains maps $f_*$ and $f^\dag$, represented by $\langle A^*,f_*,B^* \rangle$ and $\langle B,f^\dag,A \rangle$ respectively. When acting on compositions of paths, such as $\langle A,f,B,g,C,\ldots,X,h,Y,t,Z \rangle$, the dagger operator reverses the order of operations, yielding:
\[ \langle Z,t^\dag,Y,h^\dag,X,\ldots,C,g^\dag,B,f^\dag,A \rangle \]

\subsubsection{The dagger lambda calculus}
\label{Subsubsection:The dagger lambda calculus}
This section demonstrates how the graph signature $\mathcal{G}$ can be interpreted to derive the dagger lambda calculus. The set of types used by $\dag\lambda$ is precisely the set of vertices $Dagger(\Sigma_0)$ used in graph $\mathcal{G}$. Every edge $f : A \rightarrow B$ in $\Sigma_1$ is interpreted as a sequent $a:A \vdash_{\{f:a_* \otimes b\}} b:B$ up to alpha-equivalence. These interpretations essentially introduce constants, in our case $f:A^* \otimes B$, written as sequents that are reminiscent of $\eta$-expanded forms. The rest of the rules of the dagger lambda calculus can be used to process and combine sequents, yielding a richer logical structure.

\subsubsection{The syntactic category}
\label{Subsubsection:The syntactic category}
Following a method that is similar to \cite{Mel09}, we will define a process of Cut-elimination by using the soup reduction relation to partition the sequents of the dagger lambda calculus into equivalence classes. The resulting equivalence classes are modular proof invariants called \textit{denotations}. This section demonstrates how these denotations give rise to the \textit{syntactic category} $\mathcal{C}_{Synt}$, a dagger compact category. Sketches of the proofs are presented in the Appendix and more detailed versions can be found in \cite{Atz13}.

\begin{definition}[Denotations]
We will use the term \textit{denotations} to refer to the equivalence classes that are formed by partitioning the sequents of the lambda calculus according to soup equivalence. Hence, two sequents will correspond to the same denotation if and only if they are equivalent up to soup reduction.
\end{definition}

\begin{restatable}[The syntactic category]{theorem}{SyntacticCategory}
The types of the lambda calculus and the denotations generated by soup equivalence form a category whose objects are types and whose arrows are denotations.
\end{restatable}

\begin{restatable}[Dagger compact closure]{theorem}{DaggerCompactClosure}
The syntactic category is a dagger compact category.
\end{restatable}

\subsubsection{Proof of equivalence}
\label{Subsubsection:Proof of equivalence}
We will now prove that the free dagger compact category $\mathcal{C}_{Free}$ is equivalent to the syntactic category $\mathcal{C}_{Synt}$.

\begin{lemma}[Essentially surjective on objects]
The set of objects in the free category and the set of objects in the syntactic category are surjective, up to isomorphism.
\end{lemma}
\begin{proof}
Recall $Dagger(\Sigma_0)$; the free $(\otimes, I, \Box^*)$-algebra over the set of object variables $\Sigma_0$. The sets of objects in $\mathcal{C}_{Free}$ and $\mathcal{C}_{Synt}$ both correspond to $Dagger(\Sigma_0)$, up to the equivalence classes induced by $(A^*)^* \equiv A$ and $I^* \equiv I$.
\end{proof}

\begin{lemma}[Equal arrows correspond to equal denotations]
If two arrows, $\langle A,f,B \rangle$ and $\langle A,f^\prime,B \rangle$ are equal in the free category, then they will also be equal in the syntactic category: $[f] = [f^\prime] : A \rightarrow B$.
\end{lemma}
\begin{proof}
The structure of the free category $\mathcal{C}_{Free}$ imposes the minimum number of equalities for a category to be dagger compact. Moreover, both the free category and the syntactic category derive their symbols from the same signature graph $\mathcal{G}$. Since we have already shown that $\mathcal{C}_{Synt}$ is dagger compact, the same steps can be used to show that any arrows $\langle A,f,B \rangle$ and $\langle A,f^\prime,B \rangle$ that are equal in the free category correspond to equal denotations $[f] = [g]$ in the syntactic category.
\end{proof}

\begin{lemma}[Equal denotations correspond to equal arrows]
Any denotations that are equal in the syntactic category correspond to equal arrows in the free category.
\end{lemma}
\begin{proof}
Let $[f] : \Gamma \rightarrow B$ and $[g] : \Gamma \rightarrow B$ be denotations in the syntactic category such that $[f] = [g]$. Since the two denotations are equal, the sequents they represent in the dagger lambda calculus must be equivalent up to soup reduction. Without loss of generality, let's assume that $[f]$ represents a sequent $J_1$ and that $[g]$ represents a sequent $J_2$, where $J_1 \rightarrow J_2$. The soup reduction relation consists of four soup rules: \textit{bifunctoriality}, \textit{trace}, \textit{cancellation}, and \textit{consumption}. We prove this lemma by induction on the structure of the soup reduction linking $J_1$ and $J_2$. The details of the induction have been omitted in this paper; they are, available in \cite{Atz13}. This shows that $\langle \Gamma, f, B \rangle = \langle \Gamma, g, B \rangle$.
\end{proof}

\begin{theorem}[Equivalence between the free category and the syntactic category]
The free dagger compact category $\mathcal{C}_{Free}$ and the syntactic category $\mathcal{C}_{Synt}$ are equivalent.
\end{theorem}
\begin{proof}
The two categories derive their symbols from a common signature graph $\mathcal{G}$. As we have already shown, bearing in mind the equivalence classes that we have induced on types, the categories are essentially surjective on objects. Moreover, arrows that are equal in the free category are equal in the syntactic category and vice versa. This means that the functor $F$ is \textit{full} and \textit{faithful}, causing the notions of equality between arrows to overlap in these two categories. Consequently, the categories are equivalent.
\end{proof}

\begin{corollary}[Internal language]
The dagger lambda calculus is an internal language for dagger compact categories.
\end{corollary}

\section{Conclusion}
\label{Section:Conclusion}
This paper has presented a lambda calculus for dagger compact categories. As we have seen from \cite{AC04}, this language can be used to represent a subset of quantum computation, namely, quantum protocols. The dagger lambda calculus was shown to satisfy subject reduction, confluence, strong normalisation, and consistency, while the language was shown to be an internal language for dagger compact categories.

In order to be able to cover all of quantum computation, commonly referred to as universal quantum computation, we need a language with classical control. One way of adding this feature in a denotationally sound way is by extending our language's axiomatisation to include classical basis states. This can be achieved by introducing complementary classical structures, like the ones built on top of the dagger compact structure in \cite{CPP10}, \cite{CD11} and \cite{CPP08}. This work is partly covered by \cite{Atz13} and will be included in a forthcoming paper.

\section*{Acknowledgements}
\label{Section:Acknowledgements}
I would like to thank Samson Abramsky, Bob Coecke, Prakash Panangaden, Jonathan Barrett, and the anonymous reviewers for their invaluable comments and insights.

\bibliographystyle{eptcs}
\bibliography{paper}

\clearpage
\appendix
\section{Appendix}
\subsection{Correspondence to dagger compact categories}
\label{Subsection:Correspondence to dagger compact categories}
\subsubsection{The syntactic category}
\label{Subsubsection:The syntactic category}
\SyntacticCategory
\begin{proof}
As we noticed during the proof of the subject reduction property, soup reduction rules do not affect our language's type assignments. Consequently, the type of the premises used by a sequent will be the same across all sequents in a given denotation. Similarly, the type of the conclusion produced by a sequent will be the same across all sequents in a given denotation. For any sequent $\Gamma \vdash_S b:B$, corresponding to a denotation $[\pi_1]$, we will say that its \textit{domain} is $\Gamma$ and its \textit{codomain} is $B$, writing this as $[\pi_1] : \Gamma \rightarrow B$.

Let $[f]:A \rightarrow B$ and $[g]:B \rightarrow C$ be denotations representing the soup equivalent forms of some sequents $a:A \vdash_{S_1} b:B$ and $b^\prime:B \vdash_{S_2} c:C$ respectively. For any two such denotations, where the codomain of the first matches the domain of the second, we will define a \textit{composition} operator $\circ$ that can combine them into $[g] \circ [f] : A \rightarrow C$. The new denotation will represent all the soup equivalent forms of the sequent that is generated by combining the two sequents using the Cut rule:
\begin{prooftree}
    \AxiomC{$a:A \vdash_{S_1} b:B$}
    \AxiomC{$b^\prime:B \vdash_{S_2} c:C$}
    \RightLabel{\footnotesize Cut}
    \BinaryInfC{$a:A \vdash_{S_1 \cup S_2 \cup \{b : b^\prime\}} c:C$}
\end{prooftree}

The composition operation we just defined inherits associativity from the Cut rule; the order in which Cuts are performed does not matter since the connected terms are allowed to "float" freely within the soup. Therefore, $[h] \circ ([g] \circ [f]) = ([h] \circ [g]) \circ [f]$. Moreover, for every type $A$, there is a denotation $[id_A]$ that represents the sequent generated by the Identity axiom (Id): $x:A \vdash x:A$.

Composing a denotation $[f]:A \rightarrow B$ with an identity yields $[f] \circ [id_A]$ or $[id_B] \circ [f]$ depending on whether we compose with an identity on the right or on the left. The two resulting denotations represent

\begin{center}
\AxiomC{$x:A \vdash x:A$}
\AxiomC{$a:A \vdash_S b:B$}
\BinaryInfC{$x:A \vdash_{S \cup \{x : a\}} b:B$}
\DisplayProof
\hskip 0.3cm and \hskip 0.3cm
\AxiomC{$a:A \vdash_S b:B$}
\AxiomC{$x:B \vdash x:B$}
\BinaryInfC{$a:A \vdash_{S \cup \{b : x\}} x:B$}
\DisplayProof
\end{center}

\noindent both of which are soup equivalent to $a:A \vdash_S b:B$ and the rest of the sequents represented by $[f]$. Hence $[id_B] \circ [f] = [f] = [f] \circ [id_A]$
\end{proof}

\begin{definition}[Syntactic category notational conventions]
For notational convenience, we define the following combinators:\\
    \noindent $\alpha_{A,B,C} := \lambda\left(a \otimes (b \otimes c)\right).\left((a \otimes b) \otimes c\right) : \left(A \otimes (B \otimes C)\right) \multimap \left((A \otimes B) \otimes C\right)$ \\
    \noindent $\eta_A := \lambda 1.(x_* \otimes x) : I \multimap (A^* \otimes A)$ \hskip 0.7cm
    $\lambda_A := \lambda(1 \otimes a).a : (I \otimes A) \multimap A$ \hskip 0.7cm
    $\rho_A := \lambda(a \otimes 1).a : (A \otimes I) \multimap A$ \\
    \noindent $\varepsilon_A := \lambda(x \otimes x_*).1 : (A \otimes A^*) \multimap I$ \hskip 0.7cm
    $\sigma_{A,B} := \lambda(a \otimes b).(b \otimes a) : (A \otimes B) \multimap (B \otimes A)$
\end{definition}

\begin{theorem}[Monoidal category]
The syntactic category is a monoidal category.
\end{theorem}
\begin{proof}
Let $[f]:A \rightarrow B$ and $[g]:C \rightarrow D$ be denotations representing the soup equivalent forms of $a:A \vdash_{S_1} b:B$ and $c:C \vdash_{S_2} d:D$. For any such $[f]$ and $[g]$, we define a monoidal product $[f] \otimes [g] : A \otimes B \rightarrow C \otimes D$. The product represents all the soup equivalent sequents generated by using the right tensor rule to combine the sequents for $[f]$ and $[g]$. We can now use soup reduction to show that $([g] \circ [f]) \otimes ([t] \circ [h]) = ([g] \otimes [t]) \circ ([f] \otimes [h])$, $[id_A] \otimes [id_B] = [id_{A \otimes B}]$, $[\alpha_{A \otimes B, C, D}] \circ [\alpha_{A,B,C \otimes D}] = ([\alpha_{A,B,C}] \otimes [id_D]) \circ [\alpha_{A,B \otimes C,D}] \circ ([id_A] \otimes [\alpha_{B,C,D}])$, and $([\rho_A] \otimes [id_B]) \circ [\alpha_{A,I,B}] = [id_A] \otimes [\lambda_B]$. The syntactic category, therefore, satisfies all of the requirements and coherence conditions of a monoidal category.
\end{proof}

\begin{theorem}[Symmetric monoidal category]
The syntactic category is a symmetric monoidal category.
\end{theorem}
\begin{proof}
We can use soup reduction to show that $[\sigma_{B,A}] \circ [\sigma_{A,B}] = [id_{A \otimes B}]$, $[\rho_A] = [\lambda_A] \circ [\sigma_{A,I}]$, and $[\alpha_{C,A,B}] \circ [\sigma_{A \otimes B,C}] \circ [\alpha_{A,B,C}] = ([\sigma_{A,C}] \otimes [id_B]) \circ [\alpha_{A,C,B}] \circ ([id_A] \otimes [\sigma_{B,C}])$. The syntactic category thus satisfies all of the requirements and coherence conditions of a symmetric monoidal category.
\end{proof}

\begin{theorem}[Compact closure]
The syntactic category is a compact closed category.
\end{theorem}
\begin{proof}
Using our soup reduction rules, we can show that $[\lambda_A] \circ ([\varepsilon_A] \otimes [id_A]) \circ [\alpha_{A,A^*,A}] \circ ([id_A] \otimes [\eta_A]) \circ [\rho_A]^{-1} = [id_A]$ and $[\rho_{A^*}] \circ ([id_{A^*}] \otimes [\varepsilon_A]) \circ [\alpha_{A^*,A,A^*}]^{-1} \circ ([\eta_A] \otimes [id_{A^*}]) \circ [\lambda_{A^*}]^{-1} = [id_{A^*}]$, by reducing the sequents represented by the denotations on the left hand sides to identities. The syntactic category thus satisfies both of the yanking conditions that are required of a compact closed category.
\end{proof}

\DaggerCompactClosure
\begin{proof}
For every denotation $[f] : A \rightarrow B$, we define its dagger $[f]^\dag : B \rightarrow A$, as the denotation representing the soup equivalent sequents of the $\dag$-flipped sequents for $[f]$. It is now easy to show that $([f]^\dag)^\dag = [f]$ and $[\sigma_{A,A^*}] \circ [\varepsilon_A]^\dag = [\eta_A]$, by showing that the sequents they represent are soup equivalent. The syntactic category, therefore, satisfies all of the requirements of a dagger compact category.
\end{proof}

\subsection{Example}
We will examine the differences in representation between teleportation\footnote{Our analysis will not include the unitary corrections that are typically applied at the end of the teleportation protocol, as the classical control they require is beyond the scope of this paper.} of a single state and teleportation of an entire function. The "yanking" action of teleportation can be witnessed by considering the reduction:
\begin{align*}
    x_1:T & \vdash_{\{x_1 \otimes x_{2*} \otimes 1 : \varepsilon, \eta : 1 \otimes x_{2*} \otimes x_3\}} x_3:T \\
    x_1:T & \vdash_{\{x_1 \otimes x_{2*} \otimes 1 : x_4 \otimes x_{4*} \otimes 1, \eta : 1 \otimes x_{2*} \otimes x_3\}} x_3:T \\
    x_1:T & \vdash_{\{x_1:x_4, x_{2*}:x_{4*}, 1:1, \eta : 1 \otimes x_{2*} \otimes x_3\}} x_3:T \\
    x_1:T & \vdash_{\{x_{2*}:x_{1*}, \eta : 1 \otimes x_{2*} \otimes x_3\}} x_3:T \\
    x_1:T & \vdash_{\{\eta : 1 \otimes x_{1*} \otimes x_3\}} x_3:T \\
    x_1:T & \vdash_{\{1 \otimes x_{5*} \otimes x_5 : 1 \otimes x_{1*} \otimes x_3\}} x_3:T \\
    x_1:T & \vdash_{\{1:1, x_{5*}:x_{1*}, x_5:x_3\}} x_3:T \\
    x_1:T & \vdash_{\{x_1:x_3\}} x_3:T \\
    x_1:T & \vdash x_1:T
\end{align*}
For a state of type $A$, we could replace the type $T$ with $A$ and leave the rest of the sequents in the derivations as they are. Similarly, for a function of type $A \multimap B$, we could replace $T$ with $A \multimap B$ and keep the rest of the derivation intact. This reveals the power of the dagger lambda calculus; we are essentially using the same syntax to represent all types of teleportation.

\end{document}